\RequirePackage{fix-cm}
\documentclass[twocolumn,epjc3]{my_svjour3}  
\smartqed  % flush right qed marks, e.g. at end of proof
\RequirePackage{graphicx}
\RequirePackage{amsmath}
\RequirePackage{latexsym}
\newcommand{\aln}{&&\hskip -3pt}
\newcommand{\mj}{{\{m_j\}}}
\newcommand{\mjone}{{\{m_1\}}}
\newcommand{\mjtwo}{{\{m_2\}}}
\newcommand{\aem}{\alpha_{\scriptscriptstyle E\!M}}
\newcommand{\mur}{\mu_{\scriptscriptstyle R}}
\newcommand{\Lr}{L_{\scriptscriptstyle R}}

\newcommand{\bqa}{\begin{eqnarray}}
\newcommand{\eqa}{\end{eqnarray}}

\newcommand{\nl}{\nonumber \\}
\journalname{}
\begin{document}
\title{Matching high-energy electroweak fermion loops onto the Fermi theory without higher dimensional operators%\thanksref{t1}
}
\author{R. Pittau\thanksref{e1,addr1}
}

\thankstext{e1}{e-mail: pittau@ugr.es}

\institute {Departamento de F\'isica Te\'orica y del Cosmos and CAFPE, Universidad de Granada, Campus Fuentenueva s.n., E-18071 Granada, Spain \label{addr1}
}

\date{}

\maketitle

\begin{abstract}
We derive the conditions for matching high-energy renormalizable Quantum Field Theories onto low-energy nonrenormalizable ones by means of the FDR approach described in \cite{Pittau:2013ica}. Our procedure works order-by-order in the loop expansion and avoids the addition of higher dimensional interactions into the nonrenormalizable Lagrangian.
To illustrate our strategy, we match the high-energy fermion-loop corrections computed in the complete electroweak theory onto the nonrenormalizable four-fermion Fermi model.  As a result, the Fermi Lagrangian can be used without modifications to reproduce, at arbitrary loop orders and energies, the exact electroweak interactions between two massless fermion lines induced by one-fermion-loop resummed gauge boson propagators.
\end{abstract}

\section{Introduction}
Renormalizable Quantum Field Theories (QFT) are the commonly used language to describe high-energy interactions in particle physics. They are considered as fundamental theories, in the sense that predictions can be obtained, at any desired perturbative order and scale, by consistently reabsorbing the ultraviolet (UV) infinities appearing in the intermediate stages of the calculation in the set $\{p_i\}$, $i= 1\div m$, of the free parameters of the Lagrangian
\bqa
{\cal L}(p_1,\ldots,p_m).
\eqa

On the other hand, nonrenormalizable QFTs belong to a larger class of theories, namely the effective QFTs (EFT), and are extensively employed in cases when the fundamental renormalizable model is unknown, or not easily calculable.
The problem of computing high-energy loop corrections in EFTs is usually dealt with by using the seminal Weinberg's approach \cite{Weinberg:1978kz}, in which higher dimensional operators $O_i$, compatible with the symmetries of the theory, are added to the lowest order Lagrangian $\cal L$ to reabsorb the UV infinities which remain after fixing the parameters of the model,
\bqa
{\cal L} \to {\cal L} + \sum_i C_i O_i := {\cal L} + {\cal L}_{\mbox{\tiny  HD}}. \nonumber  
\eqa
By doing so order-by-order in the loop expansion, EFTs can be treated as ordinary renormalizable QFTs at the price of introducing a large set of Wilson coefficients $C_i$ (possibly, an infinite one) to be fixed by experiment. Of course, not all the $C_i$ are relevant at the energy scale under study.
As a matter of fact, if $N$ is the number of independent kinematic invariants $s_n$, one organizes the EFT as a perturbative expansion in the ratios
\bqa
\label{eq:lambda}
\lambda_n=s_n/{M^2_n},\hskip 10pt {n=1 \div N},
\eqa
where the $M_n$ are mass scales parameterizing the range of validity of the effective description \cite{Wilson:1971bg,Wilson:1971dh}. In this way, physical predictions can be obtained, order-by-order in the $\lambda_n$, in terms of a finite set of measurements.

In \cite{Pittau:2013ica} a different way to include high-energy loop corrections in nonrenormalizable QFTs is presented based on FDR \cite{Pittau:2012zd}. In FDR UV divergences are eliminated by way of a redefinition of the loop integration that does not rely on an order-by-order renormalization. Hence, UV finite quantities are directly computed without adding ${\cal L}_{\mbox{\tiny  HD}}$ to ${\cal L}$. 
The price of this is the appearance of an arbitrary renormalization
scale $\mur$. In the case of renormalizable models, the dependence on $\mur$ disappears from physical predictions ${\cal O}^{\rm TH,\ell-loop}$,
\bqa
\label{eq:indmur}
\frac{d {\cal O}^{\rm TH,\ell-loop}\big(\tilde p_1(\mur),\ldots,\tilde p_m(\mur),\mur\big) }{d \mur}= 0,
\eqa
when they are expressed in terms of the set of parameters
$\{\tilde p_i(\mur)\}$ fixed by $m$ experiments ${\cal O}_i^{\rm EXP}$ determined up to the same perturbative order $\ell$ one is working,
\bqa
\label{eq:fit}
\tilde p_i(\mur):= p_i^{\rm TH,\ell-loop}({\cal O}_1^{\rm EXP}\!,\ldots,{\cal O}_m^{\rm EXP}\!,\mur),\,i=1 \div m. \nl
\eqa
On the contrary, \eqref{eq:indmur} is not fulfilled, in general, by nonrenormalizable QFTs.
However, in the procedure of \cite{Pittau:2013ica} $\mur$ is an adjustable parameter rather than a UV cutoff, \footnote{\label{foot:1}
  This means that, at any fixed value of $\mur$,
the nonrenormalizable Lagrangian ${\cal L}$ describes a legitimate effective theory, even without adding ${\cal L}_{\mbox{\tiny  HD}}$ to it.} so that an additional measurement ${\cal O}_{m+1}^{\rm EXP}$ can be used to fix it by imposing
\bqa
\label{eq:fix}
{\cal O}_{m+1}^{\rm TH,\ell-loop}\big(\tilde p_1(\mur^\prime),\ldots,\tilde p_m(\mur^\prime),\mur^\prime\big)= {\cal O}_{m+1}^{\rm EXP}.
\eqa
After this is done, observables different from those used to determine the model,
\bqa
{\cal O}_{i}^{\rm TH,\ell-loop}\big(\tilde p_1(\mur^\prime),\ldots,\tilde p_m(\mur^\prime),\mur^\prime\big),\hskip 10pt {i > m+1},
\eqa
can be predicted and tested experimentally. If in a given range of energy
\bqa
    {\cal O}_{i}^{\rm TH,\ell-loop}\big(\tilde p_1(\mur^\prime),\ldots,\tilde p_m(\mur^\prime),\mur^\prime\big)=  {\cal O}_{i}^{\rm EXP}
\eqa
for a large class of observables ${i > m+1}$, the nonrenormalizable QFT can be used as a plausible effective model.

In this work we study under which conditions a known renormalizable
theory can be matched onto a low-energy nonrenormalizable effective model by means of the FDR approach. In this case, the matching condition \eqref{eq:fix} is replaced by
\footnote{Here and in the following, amplitudes used to fix $\mur$ are denoted by the subscript $m+1$, while the label $i > m+1$ refers to processes different from those employed to determine the Lagrangian's parameters and the renormalization scale.}
\bqa
\label{eq:fix1}
B_{m+1}^{\rm \ell-loop}(\lambda,\alpha,\mur^\prime)= A_{m+1}^{\rm \ell-loop}(\lambda,\alpha),
\eqa
where $B_{m+1}$ and $A_{m+1}$ are amplitudes computed up to the $\ell^{th}$ order in the coupling constant $\alpha$ within the nonrenormalizable and renormalizable QFT, respectively, and $\lambda$ stands for all the $N$ ratios in \eqref{eq:lambda}.  In particular, we derive the conditions to be obeyed by the coefficients of the perturbative expansion of equation \eqref{eq:fix1} for ensuring the independence of $\mur^\prime$ from kinematics.
In addition, we conjecture that, when such a $\mur^\prime$ exists, additional independent amplitudes  can be matched at $\lambda \ne 0$,
\bqa
\label{eq:matchtheo}
B_i^{\rm \ell-loop}(\lambda,\alpha,\mur^\prime)= A_i^{\rm \ell-loop}(\lambda,\alpha),\hskip 10pt i> m+1,
\eqa
if they coincide at $\lambda=0$,
\bqa
\label{eq:matchtheo0}
B_i^{\rm \ell-loop}(0,\alpha,\mur^\prime)= A_i^{\rm \ell-loop}(0,\alpha),\hskip 10pt i> m+1.
\eqa
At the present stage of our investigation we cannot prove this in general. However, it holds true when the
$A_i$ are resummed one-fermion-loop amplitudes computed in the full electroweak theory and the $B_i$ are calculated in the four-fermion Fermi model.
In such a case, if $\mur^\prime$ is fixed once for all as in \eqref{eq:fix1}, the Fermi theory reproduces, at any loop order, all the exact amplitudes describing any process involving fermion-loop mediated interactions between two massless fermions at arbitrary energy scales.
This demonstrates that realistic low-energy nonrenormalizable QFTs exist that can be consistently uplifted to higher energies by FDR without modifying their Lagrangian, at least under special classes of loop corrections. 
Conversely, if nonrenormalizable and renormalizable amplitudes can be matched with a $\mur^\prime$ independent of kinematics, the coefficients of their expansions necessarily obey the same conditions which ensure the validity of \eqref{eq:fix1}.

The structure of the paper is as follows. In section \ref{sec:fdr} we recall the essential principles of FDR. The conditions for the 
matching in \eqref{eq:fix1} are derived in section \ref{sec:matchingamp}.
Section \ref{sec:effew} describes the one-fermion-loop matching of the high-energy electroweak corrections onto the Fermi model.
Finally, the last section includes a comparison between our procedure and a customary EFT approach.

\section{FDR integration and loop functions}
\label{sec:fdr}
Here we sketch out the basic axioms of FDR with the help of a simple
one-dimensional example. The interested reader can find more details in the relevant literature \cite{Pittau:2012zd,Donati:2013iya,Pittau:2013qla,Donati:2013voa,Page:2015zca,Page:2018ljf}.

Let's assume one needs to define the UV divergent integral
\bqa
\label{eq:integral}
I = \lim_{\Lambda \to \infty}\int_0^\Lambda dx \frac{x}{x+P},
\eqa
where $P$ stands for a physical energy scale.
FDR identifies the UV divergent pieces in terms of integrands independent of $P$, dubbed FDR vacua, and  rewrites
\bqa
\label{eq:integrand}
\frac{x}{x+P}= 1-\frac{P}{x}+\frac{P^2}{x(x+P)}.
\eqa
Thus, the first term in the r.h.s. of \eqref{eq:integrand} is the vacuum responsible for the linear UV divergence, while $1/x$ generates the $\ln \Lambda$ behavior.
By definition, the linearly divergent contribution is subtracted from \eqref{eq:integral} over the full integration domain
$[0,\Lambda]$, while the logarithmic divergence over the interval
$[\mur,\Lambda]$ only. The arbitrary separation scale $\mur \ne 0$ is needed to keep a-dimensional and finite the arguments of the logarithms appearing in the subtracted and finite parts.
Thus,
\bqa
\label{eq:fdrintegral1}
I_{\scriptscriptstyle \rm FDR} := I-\lim_{\Lambda \to \infty}\!\left(\int_0^\Lambda dx-
\int_{\mur}^\Lambda dx \frac{P}{x}\right)\!= P \ln\frac{P}{\mur}.
\eqa
The advantage of this definition is twofold. Firstly, the UV cutoff $\Lambda$ is traded for $\mur$, which is interpreted as the renormalization scale. Secondly, other than logarithmic UV divergences do not contribute.
The explicit appearance of $\mur$ in the interval of integration makes the use of \eqref{eq:fdrintegral1} inconvenient in practical calculations.  An equivalent definition is obtained by adding an auxiliary unphysical scale $\mu$ to $x$, $x \to \bar x:= x+\mu$, \footnote{This replacement must be performed in both numerators and denominators of the integrated functions.} and introducing an integral operator $\int_0^\infty [dx]$ which annihilates the FDR vacua before integration. Hence,
\bqa
I_{\scriptscriptstyle \rm FDR} = \int_0^\infty [dx]  
\frac{\bar x}{\bar x+P} :=  \left.\lim_{\mu \to 0} \int_0^\infty dx
\frac{P^2}{\bar x(\bar x+P)}
\right|_{\mu = \mur}, \nonumber
\eqa
where $\mu \to 0$ is an asymptotic limit.

This strategy can be extended to more dimensions and to rational integrands depending on any number of variables, as those appearing in $\ell$-loop integrals $I^\ell_{\scriptscriptstyle \rm FDR}$. They are polynomials of degree $\ell$ in
$\ln \mur^2$, \cite{Donati:2013voa} 
\bqa
\label{eq:polLr}
I^\ell_{\scriptscriptstyle \rm FDR}= \sum_{k=0}^\ell c_k \Lr^k,\hskip 10pt \Lr:= \ln(\mur^2). 
\eqa
For instance, at one loop one has
\begin{subequations}\label{eq:loopint}
\bqa
\label{eq:loopint:1}
\aln \int [d^4 q] \frac{1}{(\bar q^2-m^2)(\bar q^2+p^2 +2 q \cdot p -m_1^2)}  \nl
\aln \hskip 10pt =I^1_{\scriptscriptstyle \rm FDR}(p^2,m^2,m^2_1) =  -i \pi^2 \int_0^1 dy \ln\frac{\chi}{\mur^2},
\\
\label{eq:loopint:2}
\aln \int [d^4 q] \frac{q^\alpha}{(\bar q^2-m^2)(\bar q^2+p^2+2 q \cdot p -m_1^2)}\nl
\aln \hskip 10pt =i \pi^2 p^\alpha \int_0^1 dy y \ln\frac{\chi}{\mur^2},
\\
\label{eq:loopint:3}
\aln \int [d^4 q] \frac{q^\alpha q^\beta}{(\bar q^2-m^2)(\bar q^2+p^2 +2 q \cdot p -m_1^2)}\nl
\aln \hskip 10pt  =\frac{i \pi^2}{2} g^{\alpha \beta} 
\int_0^1 dy \chi \left(1-\ln\frac{\chi}{\mur^2}\right) + {\cal O}(p^\alpha p^\beta),
\eqa
\end{subequations}
with $\bar q^2 := q^2-\mu^2$ and $\chi := m^2y+m^2_1(1-y)-p^2y(1-y)$.
Finally, it is important to realize that internal consistency requires $\mur$ to be independent of kinematics and identical in all loop functions. This guarantees correct cancellations when combining integrals. \footnote{For example,
the UV finite combination 
$
I^1_{\scriptscriptstyle \rm FDR}(p_1^2,m^2,m^2_1)-I^1_{\scriptscriptstyle \rm FDR}(p_2^2,m^2,m^2_1)
$
is equal to the right result,
$$
\int d^4 q
\frac{2 q \cdot (p_2-p_1)+p^2_2-p^2_1}{(q^2-m^2)((q+p_1)^2-m^2_1)((q+p_2)^2-m^2_1)},
$$
only if $\mur^2$ in \eqref{eq:loopint:1} takes the same constant value in both
$I^1_{\scriptscriptstyle \rm FDR}(p_1^2,m^2,m^2_1)$ and $I^1_{\scriptscriptstyle \rm FDR}(p_2^2,m^2,m^2_1)$.
}

\section{The conditions for matching two amplitudes}
\label{sec:matchingamp}
Our aim is determining the renormalization scale $\mur^\prime$ in
\eqref{eq:fix1}.
The all-order expansions of
$A_{m+1}$ and $B_{m+1}$ read
\begin{subequations}\label{eq:eqAB}
\bqa
\label{eq:eqABa}
A_{m+1}(\lambda,\alpha)=\aln
K(\alpha)
+K(\alpha) \sum_{j=1}^{\infty} A_{0j}^\mj \lambda^\mj  \nl
\aln + K(\alpha)\!\!\sum_{i,j=1}^{\infty} A_{ij}^\mj \alpha^i \lambda^\mj, \\
\label{eq:eqABb}
B_{m+1}(\lambda,\alpha,\mur)=\aln
K(\alpha) \nl
\aln + K(\alpha)\!\!\sum_{\substack{i,j=1\\0 \leq k \leq i}}^{\infty}\! B_{ijk}^\mj \alpha^i \lambda^\mj \Lr^k, 
\eqa
\end{subequations}
where $K(\alpha)$  is defined by the constraint
\bqa
\label{eq:lemat}
B_{m+1}(0,\alpha,\mur) = A_{m+1}(0,\alpha)= K(\alpha),
\eqa
which states that the amplitudes computed in the exact theory and the effective model coincide when $\lambda \to 0$.
$A_{0j}^\mj$, $A_{ij}^\mj$, $B_{ijk}^\mj$ are perturbative coefficients,
in which $i$ refers to the $\alpha$ expansion, whereas $j$
denotes the power degree of the products of $\lambda_n$ multiplying the coefficients.
The notation $$\mj:=(m_{j1},m_{j2},\ldots,m_{jN})$$ symbolizes an assignment of  integer numbers $m_{jn} \ge 0 $ fulfilling
\bqa
\sum_{n=1}^N m_{jn}= j,
\eqa
and a sum over all possible assignments is understood when contracting with \footnote{For instance, if $N=2$,
$
  A_{02}^\mjtwo \lambda^\mjtwo =
  A_{02}^{(2,0)} \lambda_1^2
 +A_{02}^{(0,2)} \lambda_2^2
 +A_{02}^{(1,1)} \lambda_1 \lambda_2
$. 
}
\bqa
\lambda^\mj := \prod_{n=1}^N \lambda_n^{m_{jn}}.
\eqa
The coefficients in \eqref{eq:eqAB} may involve functions of $s_n$ singular at $\lambda= 0$, such as $\ln s_n$ or ${s_n}^{-\frac{1}{2}}$, \footnote{For example, if  the $\lambda_n \to 0$ asymptotic expansion of the loop functions produces a $\sqrt{\lambda_n}$, it is rewritten as $\sqrt{\lambda_n}= \lambda_n \left(M_n {s_n}^{-\frac{1}{2}} \right)$ in \eqref{eq:eqAB}.} but \eqref{eq:lemat} requires
\bqa
A_{0j}^\mj \lambda^\mj \to 0,\,
A_{ij}^\mj \lambda^\mj \to 0,\, 
B_{ijk}^\mj \lambda^\mj \to 0 \nonumber
\eqa
when $\lambda \to 0$.
Furthermore, $B_{m+1}$ in \eqref{eq:eqABb} depends on $\lambda$ only through loop corrections, unlike $A_{m+1}$. Typically, the second term in the r.h.s. of \eqref{eq:eqABa} is generated by Taylor expanding the tree-level propagators $1/(s_n-M^2_n)$ of the exact theory, that are absent in the effective model, whose natural expansion parameters are, instead, dimensionful couplings of the type
$\alpha^a/(M_n^2)^b$ with $a,b>0$. Note also that the dependence upon $\mur$ is driven by \eqref{eq:polLr}.

Solutions to \eqref{eq:fix1} are found by replacing its two sides by
\eqref{eq:eqABa} computed with $(i \le \ell, j \le \ell)$ and
\eqref{eq:eqABb} truncated at $(i \le \ell+1,j \le \ell, k \ge i-\ell)$,
and allowing $\Lr$ in \eqref{eq:eqABb} to mix different perturbative orders,
\bqa
\label{eq:Lpert}
\Lr= \sum_{i= -1}^{\ell -1} X_i \alpha^i.
\eqa
Equating the powers of $\alpha$ and $\lambda^\mj$ gives a system of equations to be fulfilled by the unknown coefficients $X_i$. We are interested in constant solutions,
\bqa
\label{eq:musol0}
\Lr^\prime := \ln\big({\mur^\prime}^2\big)= \sum_{i= -1}^{\ell -1} X^\prime_i \alpha^i,
\eqa
in which the $X^\prime_i$ are independent of both the $\lambda_n$ and the $s_n$ contained in $A_{0j}^\mj$, $A_{ij}^\mj$, $B_{ijk}^\mj$. This requirement determines the conditions to be fulfilled by the coefficients of the two series in \eqref{eq:eqAB} to be compatible with the FDR treatment of the loop integrals outlined in section \ref{sec:fdr}. In what follows, we discuss the first two perturbative orders and delineate the structure of the general  $\ell$-loop case.

When $\ell= 1$, $\Lr= X_{-1}/\alpha+ X_0$ and the system reads
\bqa
\label{eq:sis1}
\left\{
\begin{tabular}{l}
  \!\!$A_{01}^\mjone-B_{111}^\mjone X_{-1}-B_{212}^\mjone X^2_{-1}= 0,$
  \\\\
  \!\!$A_{11}^\mjone-B_{110}^\mjone-B_{111}^\mjone X_0-B_{211}^\mjone X_{-1}$
  \\$-2 B_{212}^\mjone X_{-1} X_0 = 0, \hskip 85pt \forall \mjone$. 
\end{tabular} \right. 
\eqa
If $N=1$, only one assignment is possible, $\mjone= (1)$, and a solution compatible with \eqref{eq:sis1} can always be found for nonexceptional values of the coefficients,
\bqa
&&{\hat X}^2_{-1} B_{212}^{(1)}+ {\hat X}_{-1} B_{111}^{(1)}  -A_{01}^{(1)}= 0, \nl
&& \hat X_{0}= \frac{A_{11}^{(1)}-B_{110}^{(1)}
  -B_{211}^{(1)} \hat X_{-1}}{B_{111}^{(1)}+2 B_{211}^{(1)} \hat X_{-1}}. 
\eqa
If, in addition, this solution is such that
\bqa
\label{eq:inds1}
\frac{\partial \hat X_{i}}{\partial s_n}= 0\hskip 10pt\forall n, \hskip 5pt i=-1,0,
\eqa
then
\bqa
\label{eq:solp1}
X^\prime_{i}= \hat X_{i}, \hskip 5pt i=-1,0.
\eqa
With $N$ invariants, there are $N$ possible assignments,
$$\mjone= (1,0,\ldots,0),(0,1,\ldots,0),\ldots,(0,0,\ldots,1),$$
so that \eqref{eq:sis1} is a system of $2N$ equations and two unknowns, that admits solutions only if relations exist among the coefficients. In practice, one
determines $\hat X_{-1}$ and $\hat X_0$ for a particular assignment and checks whether this solution obeys \eqref{eq:sis1} $\forall \mjone$. After that, one
also verifies the validity of \eqref{eq:inds1}.
Thus, \eqref{eq:sis1} and \eqref{eq:inds1}
give $4N$ conditions. If they are all obeyed, the matching 
\bqa
\label{eq:1loopmatching}
B_{m+1}^{\rm 1-loop}(\lambda,\alpha,\mur^\prime)= A_{m+1}^{\rm 1-loop}(\lambda,\alpha)
\eqa
is realized by inserting \eqref{eq:solp1} in \eqref{eq:musol0} with
$\ell= 1$.

If $\ell= 2$, $\Lr= X_{-1}/\alpha+ X_0+ X_1 \alpha$, and 
\begin{subequations}{\label{eq:sis2tot}}
\bqa
\label{eq:sis2}
&&\left\{
\begin{tabular}{l}
$A_{0j}^\mj-B_{1j1}^\mj X_{-1}-B_{2j2}^\mj X_{-1}^2$\\$-B_{3j3}^\mj X_{-1}^3= 0$, \\\\
$A_{1j}^\mj-B_{1j0}^\mj-B_{1j1}^\mj X_0-B_{2j1}^\mj X_{-1}$ \\
    $-2 B_{2j2}^\mj X_{-1} X_0 -B_{3j2}^\mj X^2_{-1}$\\$-3 B_{3j3}^\mj X^2_{-1} X_0 = 0$, \\\\
$A_{2j}^\mj-B_{2j0}^\mj-B_{1j1}^\mj X_1-B_{2j1}^\mj X_0 $ \\
  $-B_{2j2}^\mj (X_0^2+2 X_{-1} X_1)$\\
  $-B_{3j1}^\mj X_{-1}-2 B_{3j2}^\mj X_{-1} X_0 $\\
  $-3B_{3j3}^\mj (X_{-1} X_0^2+X^2_{-1} X_1)= 0$, 
\end{tabular} \right. \\\nonumber \\
\label{eq:asssis2}
&&\hskip 14pt \forall \mj,~{\rm with}~j= 1 \div 2. 
\eqa
\end{subequations}
Values of $\hat X_{-1}$, $\hat X_{0}$ and $\hat X_1$ fulfilling
\eqref{eq:sis2} can in general be found for a particular assignment. Subsequently, one checks if
\bqa
\label{eq:inds2}
\frac{\partial \hat X_{i}}{\partial s_n}= 0\hskip 10pt\forall n,
\hskip 5pt i=-1 \div 1,
\eqa
and whether this very same solution holds for all the remaining assignments of \eqref{eq:asssis2}.
Therefore, 
\eqref{eq:sis2tot} and \eqref{eq:inds2} give the conditions for the matching 
\bqa
\label{eq:2loopmatching}
B_{m+1}^{\rm 2-loop}(\lambda,\alpha,\mur^\prime)= A_{m+1}^{\rm 2-loop}(\lambda,\alpha).
\eqa
If they are met, \eqref{eq:2loopmatching} is obeyed by setting
$\ell= 2$ and $X^\prime_{i}= \hat X_{i}$ in \eqref{eq:musol0}.

At $\ell$ loops and fixed assignment, $\hat X_{-1}$ is a solution of an algebraic equation of degree $(\ell+1)$. Once $\hat X_{-1}$ is known, the rest of the system is linear and triangular, so that the remaining coefficients $\hat X_{i}$,
$i= 0 \div (\ell -1)$, can be easily determined. After that, one checks the validity of this solution for all the other assignments. If, in addition,
\bqa
\frac{\partial \hat X_{i}}{\partial s_n}= 0\hskip 10pt\forall n,
\hskip 5pt i=-1 \div (\ell -1),
\eqa
the matching is achieved by choosing $X^\prime_{i}= \hat X_{i}$ in \eqref{eq:musol0}.

\section{An effective model for the high-energy electroweak fermion loops}
\label{sec:effew}
When the constraints derived in the previous section are fulfilled, the result predicted by $A_{m+1}$ is reproduced, order by order in $\alpha$ and $\lambda$, by the effective nonrenormalizable amplitude $B_{m+1}$. This allows one to determine $\mur^\prime$ from \eqref{eq:fix1} and use it in further amplitudes $B_i$ computed within the effective model.
If, after fixing the Lagrangian's parameters as in \eqref{eq:fit}, the $B_i$ obey \eqref{eq:matchtheo0}, we argue that they can matched as in
\eqref{eq:matchtheo}. Here we prove this in the case of the electroweak Fermi model when the coupling constant expansion is in terms of resummed one-fermion-loop corrections. In section \ref{sec:models} we detail the nonrenormalizable and renormalizable theories to be matched and the radiative corrections involved.
The fitting procedure of \eqref{eq:fit} is discussed
in section \ref{sec:renormalization} and the matching implied by \eqref{eq:fix1} and \eqref{eq:matchtheo} is the subject of section \ref{sec:matching}. 

\subsection{The models and the loop corrections}
\label{sec:models}
Our renormalizable theory is defined by the fermionic sector of the electroweak standard model interaction Lagrangian, namely
\bqa
\label{eq:Lren}
    {\cal L}^{\mbox{\tiny  SM}}_{\mbox{\tiny INT}}=
   {\cal L}^{\mbox{\tiny  QED}}_{\mbox{\tiny INT}}+{\cal L}^{\mbox{\tiny  ZW}}_{\mbox{\tiny INT}},
\eqa
with
 \bqa
 {\cal L}^{\mbox{\tiny  QED}}_{\mbox{\tiny INT}}=\aln
 -g s_\theta A_\alpha \sum_f Q_f \bar f_j \gamma^\alpha f_j
 \eqa
 and
 \bqa
 \label{eq:Lrenint}
    {\cal L}^{\mbox{\tiny  ZW}}_{\mbox{\tiny INT}}=\aln
    -\frac{g}{2 c_\theta} Z_\alpha\sum_f
    \bar f_j \gamma^\alpha (v_f+a_f \gamma_5) f_j \nl
    \aln
    -\frac{g}{2 \sqrt{2}} W^+_\alpha\sum_f \frac{2I_{3f}+1}{2}
    \bar f_j \gamma^\alpha (1-\gamma_5) f^\prime_j \nl
    \aln
    -\frac{g}{2 \sqrt{2}} W^-_\alpha\sum_f \frac{1-2I_{3f}}{2}
        \bar f_j \gamma^\alpha (1-\gamma_5) f^\prime_j.
\eqa
The photon and the massive gauge boson fields are denoted by
$A_\alpha$, $Z_\alpha$ and $W^{\pm}_\alpha$, respectively. The spinor associated with a fermion $f$ with color $j$ is denoted by $f_j$, with the convention that
$j= 1\div 3$ for quarks and $j= 1$ for leptons.
The sum runs over all fermions and $f^\prime$ is the isospin partner of $f$ in the limit of diagonal CKM quark-mixing matrix. 
The vector and axial couplings are
\bqa
\label{eq:Zcopulings}
v_f = I_{3f}-2 s^2_\theta Q_f,\hskip 10pt a_f= -I_{3f},
\eqa
where $I_{3f}$ is the third isospin component, $Q_f$ the electric charge and $s_\theta$ ($c_\theta$) is the sine (cosine) of the weak mixing angle.
The Feynman gauge is used, hence the gauge boson propagators read 
\bqa
\label{eq:props}
P_A^{\alpha \beta}(p^2) =\aln -i g^{\alpha \beta}\frac{1}{p^2},\hskip 10pt 
P_W^{\alpha \beta}(p^2) = -i g^{\alpha \beta}\frac{1}{p^2-M^2},\nl 
P_Z^{\alpha \beta}(p^2) =\aln -i g^{\alpha \beta}\frac{1}{p^2-M^2/c^2_\theta}.
\eqa

Our effective nonrenormalizable interaction Lagrangian reads 
\bqa
\label{eq:Lnonren}
{\cal L}^{\mbox{\tiny  EFF}}_{\mbox{\tiny INT}}=
{\cal L}^{\mbox{\tiny  QED }}_{\mbox{\tiny INT}}+{\cal L}^{\mbox{\tiny FERMI  }},
\eqa
with
\bqa
\label{eq:Lnonrenint}
 {\cal L}^{\mbox{\tiny  FERMI}}=
 -\frac{g^2}{8 M^2} {J^\dag}_{\!\!\!c\alpha} J_{c}^\alpha -\frac{g^2}{8 M^2} J_{n\alpha} J_{n}^\alpha,
 \eqa
where the charged and neutral currents are given by
 \bqa
 J_{c}^\alpha=\aln \sum_f \frac{2I_{3f}+1}{2}
 \bar f_j \gamma^\alpha (1-\gamma_5) f^\prime_j,\nl
 J_{n}^\alpha=\aln \sum_f \bar f_j \gamma^\alpha (v_f+a_f \gamma_5) f_j.
 \eqa
 In \eqref{eq:Lnonrenint} the four-fermion coupling between currents is written in a form which reproduces the tree-level low-energy result obtained with
${\cal L}^{\mbox{\tiny  SM}}_{\mbox{\tiny INT}}$ when using $P_{W,Z}^{\alpha \beta}(0)$.
Massive gauge boson propagators are absent in the effective theory, while the photon propagator is as in \eqref{eq:props}.
 
The main objects entering our calculation are the truncated one-fermion-loop contributions depicted in figure~\ref{fig:diagrams}. Fermion masses are neglected, when possible, except in the case of the top quark, for which the leading $m^2_t$ contribution is also included. The $p^\alpha p^\beta$ parts are omitted, because they do not contribute on-shell.
\begin{figure}[t]
\begin{center}
\vspace{-3.8cm}
\includegraphics[width=1.2\textwidth, angle= {0}]{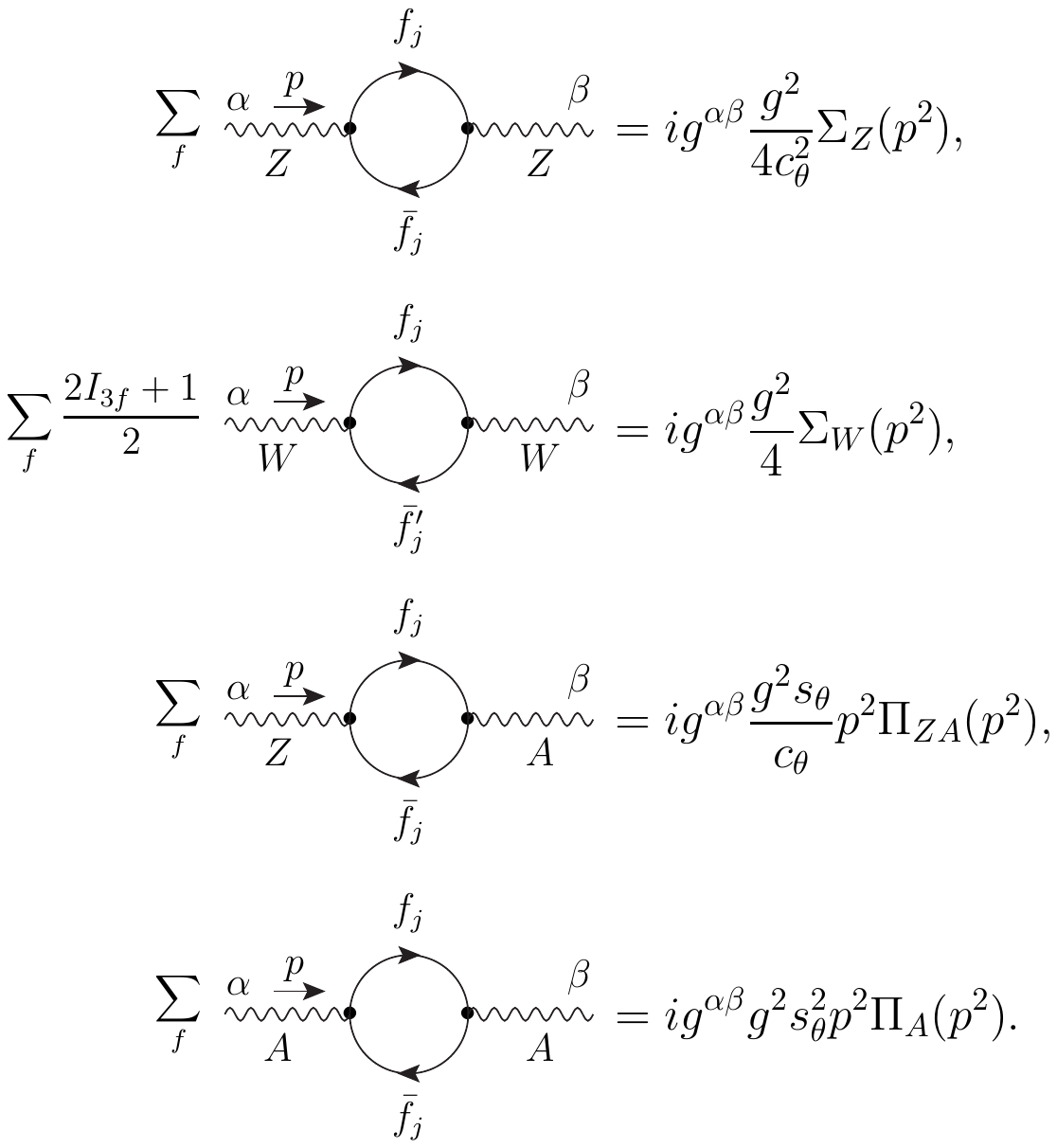}
\vspace{-2.1cm}
\caption{The parts of the truncated one-fermion-loop diagrams proportional to the metric tensor. The dots in the vertices denote that the external propagators are not included.
}
\label{fig:diagrams}
\end{center}
\end{figure}
An FDR computation of the form factors requires the integrals in \eqref{eq:loopint}. The result is
\bqa
\Sigma_Z(p^2)=\aln-\frac{p^2}{\pi^2}
\left(1-2s^2_\theta+\frac{8}{3}s^4_\theta\right)
\left(\Lr -L +\frac{5}{3}\right) \nl \aln
+ \frac{3 m^2_t}{8 \pi^2}
\left(\Lr -\ln m^2_t\right), \nl
\Sigma_W(p^2)=\aln-\frac{p^2}{\pi^2}
\left(\Lr -L +\frac{5}{3}\right) \nl \aln
+ \frac{3 m^2_t}{8 \pi^2} 
\left(\Lr -\ln m^2_t +\frac{1}{2}\right), \nl
\Pi_A(p^2)=\aln-\frac{2}{3\pi^2}
\left(\Lr -L +\frac{5}{3}\right), \nl
\Pi_A(0)=\aln-\frac{2}{3\pi^2}
\left(\Lr -K_2\right), \nl
\Pi_{Z A}(p^2)=\aln-\frac{1}{\pi^2}
\left(\frac{1}{4}-\frac{2}{3} s^2_\theta \right)
\left(\Lr -L +\frac{5}{3}\right) \nl
=\aln \Pi_{A Z}(p^2), \nl
\Pi_{Z A}(0)=\aln-s^2_\theta \Pi_A(0) -\frac{1}{4\pi^2}
\left(\Lr -K_1\right) \nl
=\aln \Pi_{A Z}(0), 
\eqa
with $L:= \ln(-p^2-i \epsilon)$.
Furthermore
\bqa
\label{eq:k1k2}
K_1:=\aln \frac{1}{2}
+\frac{\ln m^2_e +\ln m^2_\mu +\ln m^2_\tau}{12} \nl \aln
+\frac{\ln m^2_u +\ln m^2_c +\ln m^2_t}{6} \nl \aln
+\frac{\ln m^2_d +\ln m^2_s +\ln m^2_b}{12}, \nl
K_2:=\aln \frac{1}{2}
+\frac{\ln m^2_e +\ln m^2_\mu +\ln m^2_\tau}{8} \nl \aln
+\frac{\ln m^2_u +\ln m^2_c +\ln m^2_t}{6} \nl \aln
+\frac{\ln m^2_d +\ln m^2_s +\ln m^2_b}{24},
\eqa
where the light quark masses have to be considered as effective parameters adjusted to fit the dispersion integral defining the hadronic contribution to the vacuum polarization.
\subsection{Fixing the free parameters of the models}
\label{sec:renormalization}
Both Lagrangians in \eqref{eq:Lren} and \eqref{eq:Lnonren} depend on the set of bare parameters $\{g^2,M^2,s^2_\theta\}$, which need to be fixed by experiment. As input data we choose the fine structure constant $\aem$, measured in the Thomson limit of the Compton scattering, the muon decay constant $G_F$, extracted from the muon lifetime, and the ratio $R_{\scriptscriptstyle e \nu}$ between the total $e^- \nu_\mu$ and $e^- \bar \nu_\mu$ elastic cross sections at zero momentum transfer.
In the following, we determine and solve the fitting equations \cite{Veltman:1977kh,Bardin:1999ak}
linking $\{\aem,G_F,R_{\scriptscriptstyle e \nu}\}$ to
$\{g^2,M^2,s^2_\theta\}$
in both renormalizable and nonrenormalizable models.

In the renormalizable theory one constructs the fermion-loop
dressed propagators,
\bqa
\label{eq:dressprop}
D_V^{\alpha \beta}(p^2)= -i g^{\alpha \beta} \Delta_V(p^2),\hskip 8pt
V = W,Z,A,ZA,AZ,
\eqa
by Dyson resumming to all orders the self-energy contributions of figure~\ref{fig:diagrams}. The result reads
\bqa
\label{eq:dyson}
\Delta_W(p^2) =\aln \frac{1}{g^2} \frac{1}{P_W(p^2)}, \nl
\Delta_Z(p^2) =\aln \frac{1}{g^2} \frac{1}{P_Z(p^2)} \frac{1}{{\cal Z}(p^2)},\nl
p^2 \Delta_A(p^2) =\aln \frac{1}{P_A(p^2){\cal Z}(p^2)}, \nl
\Delta_{ZA}(p^2) =\aln g^2\frac{s_\theta}{c_\theta}\frac{\Pi_{ZA}(p^2)}{P_A(p^2)}
\Delta_Z(p^2) = \Delta_{AZ}(p^2),
\eqa
with
\bqa
\label{eq:dyson1}
P_W(p^2) =\aln \frac{p^2}{g^2}-\frac{M^2}{g^2}-\frac{\Sigma_W(p^2)}{4}, \nl   
P_Z(p^2)  =\aln \frac{p^2}{g^2}-\frac{M^2}{g^2c^2_\theta}-\frac{\Sigma_Z(p^2)}{4 c^2_\theta}, \nl
P_A(p^2) =\aln  1-g^2 s^2_\theta \Pi_A(p^2), \nl
{\cal Z}(p^2)=\aln
1-p^2 g^2 \frac{s^2_\theta}{c^2_\theta} \frac{\Pi^2_{ZA}(p^2)}{P_A(p^2) P_Z(p^2)}.
\eqa
Using the propagators in \eqref{eq:dressprop} to compute the 
Thomson scattering, the muon lifetime and $R_{\scriptscriptstyle e \nu}$, gives the fitting equations
\begin{subequations}\label{eq:fitting}
\bqa
\label{eq:fitting:1}
4 \pi \aem=\aln \frac{g^2 s^2_\theta}{1-g^2 s^2_\theta \Pi_A(0)}, \\
\label{eq:fitting:2}
\frac{G_F}{\sqrt{2}}=\aln \frac{g^2}{8 \left[M^2+\frac{g^2}{4} \Sigma_W(0) \right]}, \\
\label{eq:fitting:3}
R_{\scriptscriptstyle e \nu} = \aln \frac{16 S^4-12 S^2+3}{16 S^4-4 S^2+1},
\eqa
\end{subequations}
where
\bqa
S^2 := s^2_\theta
\left\{
1-\frac{g^2 \Pi_{ZA}(0)}{1-g^2 s^2_\theta \Pi_A(0)}
\right\}. \nonumber
\eqa

In the case of the nonrenormalizable model, it is easy to prove that
\begin{theorem}
  \label{th:1}
  Computing $\{\aem,G_F,R_{\scriptscriptstyle e \nu}\}$  in terms of

\noindent  $\{g^2,M^2,s^2_\theta\}$ produces the same fitting equations \eqref{eq:fitting} of the renormalizable theory.
\end{theorem}

\begin{proof}
When resumming to all orders the interactions mediated by the fermion loops, one arrives at results which have the same form of transitions induced by the dressed propagators of \eqref{eq:dyson} computed at $p^2= 0$. Since the observables used as input data only involve zero momentum transfer, the equations \eqref{eq:fitting} are also valid in the nonrenormalizable theory.
\end{proof}
As an example,  the diagram  relevant in the case of charged currents is given in figure~\ref{fig:wloop}. That modifies the muon decay amplitude as depicted in figure~\ref{fig:mudecay}.
One computes
\bqa
\label{eq:A0}
A^{\mbox{\tiny EFF}}_W(0)=  -\frac{i\Gamma}{8} \frac{g^2}{M^2}
\frac{1}{1+\frac{g^2}{4 M^2} \Sigma_W(0)}
=   \frac{i \Gamma}{8} g^2 \Delta_W(0), \nl
\eqa
where $\Gamma$ is the result of the contraction of the two charged currents
$
\Gamma := \gamma_\alpha (1-\gamma_5) \otimes \gamma^\alpha(1-\gamma_5)
$,
in which the symbol $\otimes$ understands multiplication by the relevant external spinors.
Using  \eqref{eq:A0} to define the combination $g^2/M^2$ leads to \eqref{eq:fitting:2}.

\begin{figure}[t]
  \vskip -68pt
  \hskip -110pt
\includegraphics[width=0.9\textwidth, angle= {-90}]{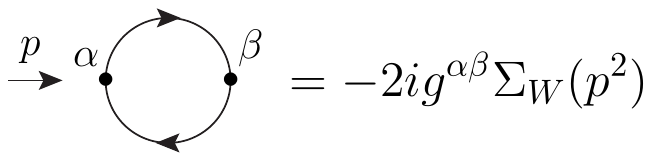}
  \vskip -330pt
\caption{The diagram mediating fermion-loop induced interactions between charged currents in the nonrenormalizable theory.}
\label{fig:wloop}
\end{figure}

\begin{figure}[t] 
  \vskip -68pt
  \hskip -100pt
\includegraphics[width=0.9\textwidth, angle= {-90}]{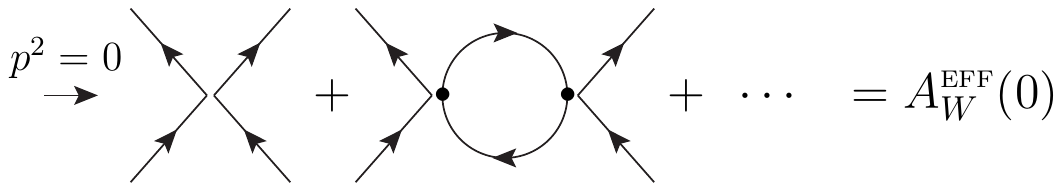}
  \vskip -325pt
\caption{The amplitude describing the muon decay in the nonrenormalizable theory.  The interaction of figure~\ref{fig:wloop} is evaluated at $p^2= 0$ and resummed to all orders.}
\label{fig:mudecay}
\end{figure}

\vspace{12pt}
Finally, to solve the fitting equations we first introduce the tree-level solution to \eqref{eq:fitting:3}, namely $\hat s_\theta$ such that
\bqa
R_{\scriptscriptstyle e \nu} = \frac{16 {\hat s_\theta}^4-12 {\hat s_\theta}^2  +3}{16 {\hat s_\theta}^4-4 {\hat s_\theta}^2+1}.
\eqa
Radiative corrections do not change $R_{\scriptscriptstyle e \nu}$ when
$S^2= {\hat s_\theta}^2$, that gives 
\begin{subequations}\label{eq:sol}
  \bqa
 \label{eq:sol:1}
s^2_\theta =\aln {\hat s_\theta}^2 \frac{F_1}{F_2}, \\
 \label{eq:sol:2}
g^2       =\aln \frac{4 \pi \aem}{{\hat s_\theta}^2F_1}, \\
 \label{eq:sol:3}
M^2       =\aln \frac{{\hat M}^2}{F_1} \left(1- \sqrt{2}G_F \Sigma_W (0)\right),
\eqa
\end{subequations}
with
\bqa
\label{eq:def}
{\hat M}^2:=\aln \frac{\pi \aem}{\sqrt{2} G_F {\hat s_\theta}^2},\hskip 10pt
F_1 := 1 -\frac{\aem}{\pi {\hat s_\theta}^2} \left(\Lr-K_1 \right),\nl 
F_2 :=\aln 1 -\frac{8\aem}{3 \pi} \left(\Lr-K_2 \right).
\eqa

\subsection{Matching the exact theory onto the nonrenormalizable model}
\label{sec:matching}
The high-energy fermion-loop corrections computed with ${\cal L}^{\mbox{\tiny  SM}}_{\mbox{\tiny INT}}$ are matched onto ${\cal L}^{\mbox{\tiny  EFF}}_{\mbox{\tiny INT}}$ by comparing amplitudes induced by charged currents of virtuality $p^2$.
In the renormalizable theory one has
\bqa
\label{eq:awsmp}
A_W^{\mbox{\tiny SM}}(p^2) =\aln
\frac{i \Gamma}{8}  g^2 \Delta_W(p^2) \nl
=\aln
\frac{i \Gamma}{8}
\left\{\frac{p^2}{g^2}-\frac{M^2}{g^2}-\frac{\Sigma_W(p^2)}{4} 
\right\}^{-1}, 
\eqa
while resumming the interaction as in figure~\ref{fig:mudecay}, but with $p^2 \ne 0$, gives
\bqa
\label{eq:aweffp}
A_W^{\mbox{\tiny EFF}}(p^2,\Lr)=
\frac{i\Gamma}{8}
\left\{-\frac{M^2}{g^2}-\frac{\Sigma_W(p^2)}{4}
\right\}^{-1}.
\eqa
Equations \eqref{eq:awsmp} and \eqref{eq:aweffp} differ by the term $p^2/g^2$, so that inserting the solution \eqref{eq:sol} produces a result independent of $\Lr$ for $A_W^{\mbox{\tiny SM}}$, whilst $A_W^{\mbox{\tiny EFF}}$ still depends on $\Lr$, 
\bqa
\label{eq:asm}
\frac{A_W^{\mbox{\tiny SM}}(p^2)}{K(\aem)} = \aln
\Bigg\{1-\frac{p^2}{{\hat M}^2} \nl
\aln -\frac{\aem}{\pi {\hat s_\theta}^2{\hat M}^2 } p^2
\left(K_1-L + {5}/{3}\right)
\Bigg\}^{-1}, \\
\label{eq:aeff}
\frac{A_W^{\mbox{\tiny EFF}}(p^2,\Lr)}{K(\aem)}= \aln
\Bigg\{1 -\frac{\aem}{\pi {\hat s_\theta}^2{\hat M}^2 } p^2
\left(\Lr-L + {5}/{3}\right)
\Bigg\}^{-1} \\
K(\aem) = \aln 
-\frac{i\Gamma}{2} \frac{\pi \aem}{{\hat s_\theta}^2{\hat M}^2}.
\nonumber
\eqa
At fixed $\ell$, the amplitudes in \eqref{eq:asm} and \eqref{eq:aeff} are the right and left sides of the matching equation \eqref{eq:fix1} needed to determine $\mur^\prime$. For instance, the conditions ensuring the validity of \eqref{eq:2loopmatching} can be verified by expanding up to the second order in
$\lambda = {p^2}/{{\hat M}^2}$,
\bqa
\label{eq:12exp}
\aln\frac{A_W^{\mbox{\tiny SM}}(p^2)}{K(\alpha)}= 
1
+\lambda \left(1 
+\frac{\alpha}{\pi {\hat s_\theta}^2} \left(K_1-L + {5}/{3}\right)\right) \nl
\aln \hskip 20pt +\lambda^2 \left(1 
+\frac{\alpha}{\pi {\hat s_\theta}^2} \left(K_1-L + {5}/{3}\right)\right)^2 
+ {\cal O}(\lambda^3), \nl
\aln \frac{A_W^{\mbox{\tiny EFF}}(p^2,\Lr)}{K(\alpha)}=  
1
+\frac{\alpha \lambda}{\pi {\hat s_\theta}^2}
\left(\Lr-L + {5}/{3}\right) \nl
\aln \hskip 20pt +\frac{\alpha^2 \lambda^2}{\pi^2 {\hat s_\theta}^4}
\left(\Lr-L + {5}/{3}\right)^2 +{\cal O}(\lambda^3), 
\eqa
where $\alpha= \aem$. From \eqref{eq:12exp} one reads off the nonzero coefficients \footnote{Since $N=1$, $\mj=(j)$.} 
\bqa
\hskip -5pt
\begin{tabular}{lll} 
  $A^{(1)}_{01}=1$, & \hskip -9pt $A^{(2)}_{02}=1$, & \hskip -9pt $A^{(1)}_{11}=\frac{5/3-L+K_1}{\pi {\hat s_\theta}^2}$,
  \hskip -9pt \\
  $A^{(2)}_{12}= 2 A^{(1)}_{11}$, & \hskip -9pt
  $A^{(2)}_{22}= \left(A^{(1)}_{11}\right)^2\!$, & \hskip -9pt
  $B^{(1)}_{110}=\frac{5/3-L}{\pi {\hat s_\theta}^2}$,
  \hskip -9pt \\
  $B^{(2)}_{220}= \left(B^{(1)}_{110}\right)^2\!$, & \hskip -9pt
  $B^{(1)}_{111}=\frac{1}{\pi {\hat s_\theta}^2}$, & \hskip -9pt
  $B^{(2)}_{221}=\frac{2}{\pi {\hat s_\theta}^2} B^{(1)}_{110} $,
  \hskip -9pt \\
  $B^{(2)}_{222}= \left(B^{(1)}_{111}\right)^2\!$, & \hskip -9pt & \hskip -9pt
\end{tabular} 
\eqa
and the solution $X^\prime_{-1}= \pi {\hat s_\theta}^2$, 
$X^\prime_{0}= K_1$, $X^\prime_{1}= 0$, namely
\bqa
\label{eq:musol}
\Lr^\prime= \frac{\pi{\hat s_\theta}^2}{\aem}+K_1,
\eqa
fulfills, for any value of $j$, all conditions stated by \eqref{eq:sis2}
and \eqref{eq:inds2}. As a matter of fact, $\Lr= \Lr^\prime$ solves \eqref{eq:fix1} to all orders. In fact, this is the value for which the resummed amplitudes of \eqref{eq:asm} and \eqref{eq:aeff} coincide. Hence, choosing the renormalization scale as in \eqref{eq:musol} reproduces the effect of interchanging a one-fermion-loop dressed $W$ boson of arbitrary virtuality $p^2$.

Now we consider a further amplitude $A_Z^{\mbox{\tiny EFF}}$ obtained by contracting two neutral currents. It obeys
\eqref{eq:matchtheo0} by construction and
\begin{theorem}
  \label{th:2} When computed at $\Lr = \Lr^\prime$, any effective amplitude involving two massless neutral currents reproduces, at any value of $p^2$, the exact all-order result predicted by ${\cal L}^{\mbox{\tiny  SM}}_{\mbox{\tiny INT}}$.
\end{theorem}
So that, $A_Z^{\mbox{\tiny EFF}}$ fulfills \eqref{eq:matchtheo} at any $\ell$.
\begin{proof}
Consider the full amplitude
\bqa
\label{eq:smnfull}
A_Z^{\mbox{\tiny SM}}(p^2)= \sum_{k=1}^4  A_k^{\mbox{\tiny SM}}(p^2,\Lr)
\eqa
describing the interaction between two massless fermions $f_1$ and $f_2$ in the renormalizable theory. A computation of the sub-amplitudes in
figure~\ref{fig:subamplitudes} gives
\bqa
\label{eq:subamplitudes}
A_1^{\mbox{\tiny SM}}(p^2,\Lr)=\aln i g^2 s^2_\theta  Q_{f_1} Q_{f_2} \Delta_A(p^2)
\gamma_\alpha \otimes
\gamma^\alpha, \nl
A_2^{\mbox{\tiny SM}}(p^2,\Lr)=\aln i g^2 \frac{1}{4 c^2_\theta}  \Delta_Z(p^2)
\gamma_\alpha(v_{f_1}+a_{f_1} \gamma_5) \nl \aln
\otimes \,\gamma^\alpha(v_{f_2}+a_{f_2} \gamma_5), \nl
A_3^{\mbox{\tiny SM}}(p^2,\Lr)=\aln  i g^2 \frac{s_\theta Q_{f_2}}{2 c_\theta} \Delta_{ZA}(p^2)
\gamma_\alpha(v_{f_1}+a_{f_1} \gamma_5) \otimes
\gamma^\alpha, \nl
A_4^{\mbox{\tiny SM}}(p^2,\Lr)=\aln i g^2 \frac{s_\theta Q_{f_1}}{2 c_\theta}  \Delta_{ZA}(p^2)
\gamma_\alpha \otimes
\gamma^\alpha(v_{f_2}+a_{f_2} \gamma_5). \nl
\eqa
\begin{figure}[t]
\vskip -55pt  
\includegraphics[width=1.2\textwidth]{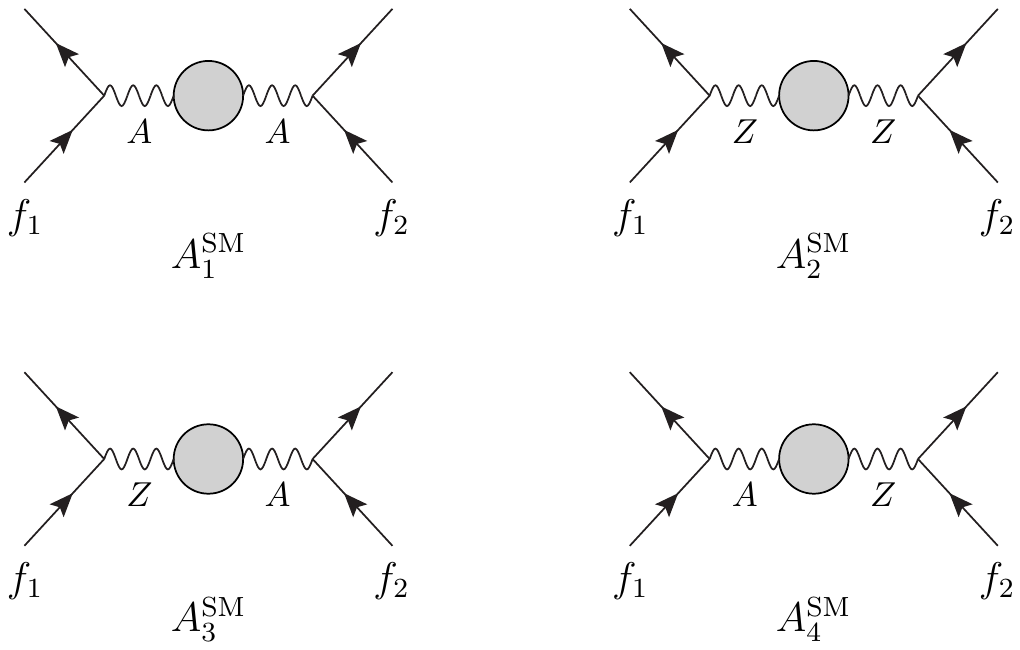}
\vskip -215pt  
\caption{The four sub-amplitudes in \eqref{eq:subamplitudes} induced by the fermion-loop dressed propagators of \eqref{eq:dressprop}. The external fermions are massless, so that diagrams involving the exchange of neutral scalars are absent.}
\label{fig:subamplitudes}
\end{figure}
Since ${\cal L}^{\mbox{\tiny  SM}}_{\mbox{\tiny INT}}$ is renormalizable, $A_Z^{\mbox{\tiny SM}}(p^2)$ does not depend on $\Lr$.  Therefore, one is allowed to choose
$\Lr= \Lr^\prime$ in each of the four sub-amplitudes. But this implies $F_1= 0$
in \eqref{eq:def}, which means ${p^2}/{g^2}= 0$ inside the function $P_Z(p^2)$ contained in the definition of the dressed propagators $\Delta_A(p^2)$, $\Delta_Z(p^2)$ and  $\Delta_{ZA}(p^2)$ in \eqref{eq:subamplitudes}. Since this is the only difference between the results computed within the nonrenormalizable and renormalizable models,
%\footnote{As in the charged current sector, cfr. \eqref{eq:awsmp} and \eqref{eq:aweffp}.}
one obtains
\bqa
\label{eq:aeffn}
A_k^{\mbox{\tiny EFF}}(p^2,\Lr^\prime)= A_k^{\mbox{\tiny SM}}(p^2,\Lr^\prime)\hskip 10pt
\forall k. 
\eqa
Thus,
\bqa
A_Z^{\mbox{\tiny EFF}}(p^2,\Lr^\prime)=\aln
 \sum_{k=1}^4 A_k^{\mbox{\tiny EFF}}(p^2,\Lr^\prime) \nl 
=\aln\sum_{k=1}^4 A_k^{\mbox{\tiny SM}}(p^2,\Lr^\prime)
= A_Z^{\mbox{\tiny SM}}(p^2).
\eqa
\end{proof}
An interesting consequence is
\begin{corollary}
 \label{cor:cor1}
In the renormalizable theory of \eqref{eq:Lren} it is possible to rearrange the fermion-loop corrections in such a way that all fermions couple to $Z$ and $W$ bosons with the same V-A interaction.
\end{corollary}
\begin{proof}
This is again obtained by choosing $\mur$ in \eqref{eq:subamplitudes} as in \eqref{eq:musol}, that implies
$s_\theta^2= 0$ in \eqref{eq:sol:1} and $v_f= -a_f = I_{3f}$ in \eqref{eq:Zcopulings}. 
\end{proof}

To summarize, any exact amplitude, in which two massless fermion lines are connected by a one-fermion-loop dressed $W$, $Z$ or $\gamma$ propagator of arbitrary
virtuality, is reproduced by ${\cal L}^{\mbox{\tiny  EFF}}_{\mbox{\tiny INT}}$ if the solution in \eqref{eq:musol} is used for the renormalization scale.

 Finally, it should be explicitly noticed that the choice of the interactions included in \eqref{eq:Lnonren} is ultimately driven by the requirement that the effective and the exact model coincide, when $\lambda \to 0$, for the class of processes and corrections under study. For example, ${\cal L}^{\mbox{\tiny  EFF}}_{\mbox{\tiny INT}}$ is too poor to accommodate contributions not induced by fermion loops, e.g. the amplitudes $B_i$ in the l.h.s. of \eqref{eq:matchtheo0} would not match the $A_i$ if the latter would involve three-gauge-boson vertices.

\section{Comparing with customary calculations}
In what follows, we use the model of \eqref{eq:Lnonren} to compare our treatment with a more standard order-by-order renormalization approach based on Dimensional Regularization (DReg). Our formulae are converted to DReg by replacing \cite{Gnendiger:2017pys}
\bqa
\label{eq:FDRvsDREG}
\Lr \to \Lr+\frac{1}{\epsilon_{\mbox{\tiny  UV}}},
\eqa
where
\bqa
\frac{1}{\epsilon_{\mbox{\tiny  UV}}} := \frac{2}{4-d}-\gamma_E -\ln \pi
\hspace{7pt} {\rm with} \hspace{10pt} d \to 4.
\eqa
Upon this substitution, the effective amplitudes in \eqref{eq:aeff} and \eqref{eq:aeffn} develop a dependence on the UV cutoff ${1}/{\epsilon_{\mbox{\tiny  UV}}}$.
To cancel it in the Weinberg's way, one adds to the effective Lagrangian interactions induced by higher dimensional operators,
\bqa
\label{eq:LCT}
 {\cal L}_{\mbox{\tiny  HD}}=\aln
 -c_w \frac{g^4}{32 M^4}(\partial_\nu {J}_{\!c\alpha})^\dag (\partial^\nu J_{c}^\alpha) \nl \aln
 -c_z \frac{g^4 c^2_\theta}{32 M^4}(\partial_\nu J_{n\alpha}) (\partial^\nu J_{n}^\alpha).
 \eqa
 Matching the exact results of \eqref{eq:asm} and \eqref{eq:smnfull} onto a computation performed with
${\cal L}^{\mbox{\tiny  EFF}}_{\mbox{\tiny INT}}+{\cal L}_{\mbox{\tiny  HD}}$ fixes the unknown coefficients,
\bqa
 c_w(\Lr) = c_z(\Lr)= \frac{{\hat s_\theta}^2 }{\pi \aem }+
 \frac{1}{\pi^2} \left(K_1-\Lr-\frac{1}{\epsilon_{\mbox{\tiny  UV}}} \right). \nonumber
 \eqa
 Even when choosing $\mur$ as in \eqref{eq:musol} only the finite parts of $c_{w,z}$ are removed,
 \bqa
 c_{w,z}(\Lr^\prime)= -\frac{1}{\pi^2\epsilon_{\mbox{\tiny  UV}}},
 \eqa
 hence adding ${\cal L}_{\mbox{\tiny  HD}}$ to 
 ${\cal L}^{\mbox{\tiny  EFF}}_{\mbox{\tiny INT}}$ is necessary to compensate the UV poles contained  in the DReg variant of the one-loop functions of \eqref{eq:loopint}. Such poles are absent when defining UV divergent integrals as in \eqref{eq:fdrintegral1}. This explains why FDR circumvents the introduction of the counterterm Lagrangian ${\cal L}_{\mbox{\tiny  HD}}$, which is instead needed in the standard method.
 \footnote{Note that FDR is not equivalent to DReg in which the loop integrals are redefined by dropping  $1/\epsilon_{\mbox{\tiny  UV}}$ terms. For instance, \cite{Donati:2013voa,tHooft:1973wag} when $\ell > 1$
 \bqa
 \aln
 {\rm Finite\,Part}
 \left\{
\int {d^dq}\, \frac{\mur^{(4-d)}}{(q^2-m^2)((q+p)^2-m^2_1)}
\right\}^\ell
\nl
 \aln \hskip 15pt \ne
 \left(
 I^1_{\scriptscriptstyle \rm FDR}(p^2,m^2,m^2_1)
  \right)^\ell, \nonumber
 \eqa
 with $I^1_{\scriptscriptstyle \rm FDR}(p^2,m^2,m^2_1)$ given in \eqref{eq:loopint:1}.
 In DReg this mismatch is cured by the $1/\epsilon_{\mbox{\tiny  UV}}$ pole contained in ${\cal L}_{\mbox{\tiny  HD}}$. Hence, setting ${\cal L}_{\mbox{\tiny  HD}}=0$ would give a wrong DReg result for the resummed propagators of \eqref{eq:dyson}.
 }
 It is also interesting to speculate about the FDR matching of \eqref{eq:musol} from the point of view of the sole EFT. In particular, would it be possible to guess the ``right'' value of $\mur$ without knowing ${\cal L}^{\mbox{\tiny  SM}}_{\mbox{\tiny INT}}$? Requiring that ${\cal L}^{\mbox{\tiny  EFF}}_{\mbox{\tiny INT}}$ describe as many processes as possible leads to the universal V-A interaction realized by the value $s^2_\theta= 0$ implied by \eqref{eq:musol}, as noted in corollary \ref{cor:cor1}. More than that, choosing $s^2_\theta= 0$ effectively reduces from three to two the number of free parameters in \eqref{eq:sol}.  In summary, minimality could be used as a criterion to fix $\mur$ in nonrenormalizable QFTs whose UV completion is unknown. Note that, in any standard procedure based on DReg, $s^2_\theta$ would be a bare parameter containing ${1}/{\epsilon_{\mbox{\tiny  UV}}}$ poles, which cannot be compensated by any finite value of $\mur$. Thus, setting $s^2_\theta= 0$ directly in \eqref{eq:sol} would not be possible.

 In the rest of this section we briefly outline the steps towards a possible generalization of our approach beyond the simple model of \eqref{eq:Lnonren}. Given the current interest in precise EFT analyses of collider data, we directly focus on a phenomenologically relevant problem by studying how new physics effects could be parameterized within the FDR framework at the NLO accuracy. \footnote{This means including all corrections ${\cal O}(g^2)$, ${\cal O}(\lambda_n)$ and ${\cal O}(g^2 \lambda_n)$ with respect to the lowest order standard model predictions.} To be definite, we consider the Lagrangian
\bqa
\label{eq:L6}
{\cal L}_{\mbox{\tiny NP}}= {\cal L}^{(4)}_{\mbox{\tiny SM}}+\frac{g^2}{\Lambda^2}{\cal L}^{(6)},
\eqa
where ${\cal L}^{(4)}_{\mbox{\tiny SM}}$ is the full standard model bare Lagrangian and $g$ is the ${\rm SU(2)_L}$ coupling constant.
${\cal L}^{(6)}$ contains a set of gauge invariant dimension-six operators, multiplied by Wilson coefficients, which we want to determine, and $\Lambda$ is the new physics scale, with which all the $M_n$ in \eqref{eq:lambda} are identified. The reader should be aware of the fact that a systematic and detailed treatment of this problem is far beyond our scope. Here we simply want to point out the general qualitative differences with respect to more standard approaches. 

Equation \eqref{eq:L6} very much resembles the customary SMEFT \cite{Brivio:2017vri} dimension-six parameterization. However, in our case the operators in ${\cal L}^{(6)}$ are not necessarily closed under renormalization. For instance, they could be a sub-set of the operators of the Warsaw basis \cite{Grzadkowski:2010es}. Furthermore, ${\cal L}_{\mbox{\tiny NP}}$ remains the same at all loop orders (see footnote \ref{foot:1}).
Before starting the calculation, one needs to expand ${\cal L}^{(6)}$ around the Higgs vacuum expectation value $v$.
This gives rise to powers of $v/\Lambda$ that modify the relations connecting 
weak eigenstates to mass eigenstates and alter the gauge fixing needed to quantize ${\cal L}_{\mbox{\tiny NP}}$.
An analogous problem is encountered in the SMEFT, and can be solved, for instance, as described in \cite{Helset:2018fgq}. \footnote{Alternatively, since our matching conditions only involve physical amplitudes, one can use any gauge expressed in terms of the bare fields in ${\cal L}^{(4)}_{\mbox{\tiny SM}}$, at the price of correcting the external particle wave functions such that propagators have residue one at their poles \cite{Passarino:2016saj}.} A difference arises when the $v/\Lambda$ terms generate contact interactions not present in
${\cal L}^{(4)}_{\mbox{\tiny SM}}$. In this case they should be included in the factor $K(\alpha)$ of \eqref{eq:lemat}. This is due to the fact that the expansion in \eqref{eq:eqABb} is in terms of the $\lambda_n$. 
 
The starting point to determine the Wilson coefficients and $\Lambda$ is a set of observables
${\cal O}_i$, $i \ge m+1$, for which there is an experimental agreement, when
all $\lambda_n \to 0$,  with the theoretical predictions obtained with ${\cal L}_{\mbox{\tiny NP}}$. \footnote{Adding real corrections might be needed at this stage to define infrared safe quantities.}
This may require to fit different compositions of the dimension-six operators in ${\cal L}^{(6)}$ until this agreement is reached.
After this is achieved, one measures one of the observables, say
${\cal O}_{m+1}$, at small values of the $\lambda_n$ and tries to determine $X^\prime_{-1,0}$ in \eqref{eq:musol0} such that the agreement persist also when $\lambda_n \ne 0$. Note that, when several $\lambda_n$ are involved, this may require measuring
${\cal O}_{m+1}$ in different phase-space regions.
If the $\lambda_n \ne 0$ agreement is not reached, one is led to reconsider once again the combination of dimension-six operators in ${\cal L}^{(6)}$.
When $X^\prime_{-1}$ and $X^\prime_0$ can be found, the theory is fixed and our conjecture states that all the other observables ${\cal O}_{i}$, $i > m+1$, are also reproduced by ${\cal L}_{\mbox{\tiny NP}}$. If necessary, this can be checked experimentally.

\section{Conclusion}
We have derived the order-by-order conditions which have to be fulfilled by effective amplitudes computed in FDR to reproduce exact high-energy predictions.
In our procedure the Lagrangian of the effective model is not modified by the inclusion of higher dimensional operators.
At the core of our analysis lies an expansion of the renormalization scale $\mur$ that mixes different perturbative orders.

We have postulated that if there exist classes of amplitudes for which the effective and the exact theory coincide at low energies, and if a value of $\mur$ can be found, for one of them, that matches at higher energies the exact result onto the effective one, all the other effective amplitudes computed at $\mur$ reproduce the exact high-energy predictions. 

We have proven this explicitly to all loop orders by matching onto the Fermi model electroweak processes induced by the exchange of a one-fermion-loop dressed $W$, $Z$ or $\gamma$ propagator of arbitrary virtuality. In such a situation our approach is more direct than a standard EFT calculation, and gives some hints on how to handle nonrenormalizable models when more fundamental theories are not known.

We plan to corroborate our conjecture by considering further classes of
theories and corrections in future investigations.

\begin{acknowledgements}
I acknowledge the financial support of the MINECO project FPA2016-78220-C3-3-P
and the hospitality of the CERN TH department during the completion of this work. I also thank Giampiero Passarino for informative discussions on the SMEFT.
\end{acknowledgements}

\bibliographystyle{spphys}       % APS-like style for physics
\bibliography{fdrmatching}   % name your BibTeX data base

\end{document}